\documentclass[final,1p,times]{elsarticle}

\usepackage{fullpage}
\usepackage{epsfig,url}
\usepackage{xspace,subfigure}
\usepackage{algorithm}
\usepackage{amsthm}
\usepackage{algpseudocode}
\usepackage{verbatim}
\usepackage{latexsym}
\usepackage{amsmath}
\usepackage{amsfonts}
\usepackage{amssymb}
\usepackage{graphicx}
\usepackage{textcomp}
\usepackage{todonotes}
\usepackage{mathtools}
\usepackage{caption}

\newtheorem{theorem}{Theorem}
\newtheorem{lemma}{Lemma}
\newtheorem{corollary}{Corollary}

\newtheorem{lem}  {Lemma} 
\newcommand {\BL} {\begin{lem}} 
\newcommand {\EL} {\end{lem}} 

\newtheorem{cor}  {Corollary} 
\newcommand {\BCR} {\begin{cor}}
\newcommand {\ECR} {\end{cor}}

\newtheorem{thm} {Theorem} 
\newcommand {\BT} {\begin{thm}}
\newcommand {\ET} {\end{thm}}

\newtheorem{defi} {Problem} 
\newcommand {\BDE} {\begin{defi}}
\newcommand {\EDE} {\end{defi}}

\newtheorem{xxx} {Definition} 
\newcommand {\BD} {\begin{xxx}}
\newcommand {\ED} {\end{xxx}}

\newtheorem{xxxxx} {Observation} 
\newcommand {\BO} {\begin{xxxxx}}
\newcommand {\EO} {\end{xxxxx}}

\journal{Elsevier}

\begin{document}

\begin{frontmatter}

\title{Biconnectivity, $st$-numbering and other applications of DFS using $O(n)$ bits  \tnoteref{t1}}
\tnotetext[t1]{Some of these results were announced in preliminary form in the proceedings of 27th International Symposium on Algorithms and Computation (ISAAC 2016) LIPIcs, Volume 64, pages 22:1-22:13~\cite{Chakraborty0S16}.}
\author[VR]{Sankardeep Chakraborty}
\ead{sankardeep@imsc.res.in}
\author[VR]{Venkatesh Raman}
\ead{vraman@imsc.res.in}
\author[SR]{Srinivasa Rao Satti}
\ead{ssrao@cse.snu.ac.kr}
\address[VR]{The Institute of Mathematical Sciences, HBNI, Chennai, India}
\address[SR]{Seoul National University, Seoul, South Korea}


\begin{abstract}
We consider space efficient implementations of some classical applications of DFS including the problem of testing biconnectivity and $2$-edge connectivity, finding cut vertices and cut edges, computing chain decomposition and $st$-numbering of a given undirected graph $G$ on $n$ vertices and $m$ edges. Classical algorithms for them typically use DFS and some $\Omega (\lg n)$ bits\footnote{We use $\lg$ to denote logarithm to the base $2$.} of information at each vertex. 
Building on a recent $O(n)$-bits implementation of DFS due to Elmasry et al. (STACS 2015) we provide 
$O(n)$-bit implementations for all these applications of DFS. Our algorithms take $O(m \lg^c n \lg\lg n)$ time for some small constant $c$ (where $c \leq 2$). Central to our implementation is a succinct representation of the DFS tree and a space efficient partitioning of the DFS tree into connected subtrees, which maybe of independent interest for designing other space efficient graph algorithms.

\end{abstract}
\end{frontmatter}

\section{Introduction}
Space efficient algorithms are becoming increasingly important owing to their applications in the presence of rapid growth of ``big data'' and the proliferation of specialized handheld devices and embedded systems that have a limited supply of memory. Even if mobile devices and embedded systems are designed with large supply of memory, it might be useful to restrict the number of write operations. For example, on flash memory, writing is a costly operation in terms of speed, and it also reduces the reliability and longevity of the memory. Keeping all these constraints in mind, it makes sense to consider algorithms that do not modify the input and use only a limited amount of work space. One computational model that has been proposed in algorithmic literature to study space efficient algorithms, is the read-only memory (ROM) model. 
In this article, we focus on space efficient implementations of some fundamental graph algorithms in such settings without paying too much penalty on time.

There is already a rich history of designing space efficient algorithms in the read-only memory model. The complexity class {\sf L} (also known as {\sf DLOGSPACE}) is the class containing decision problems that can be solved by a deterministic Turing machine using only logarithmic amount of work space for computation. There are several important algorithmic results~\cite{DattaLNTW09,ElberfeldJT10,ElberfeldK14,ElberfeldS16} for this class, the most celebrated being Reingold's method~\cite{Reingold08} for checking {\it st}-reachability in an undirected graph, i.e., to determine if there is a path between two given vertices $s$ and $t$. {\sf NL} is the non-deterministic analogue of {\sf L} and it is known that the {\it st}-reachability problem for {\it directed} graphs is {\sf NL}-complete (with respect to log space reductions). Using Savitch's algorithm~\cite{AroraB}, this problem can be solved in $n^{O(\lg n)}$ time using $O(\lg ^2 n)$ bits. Savitch's algorithm is very space efficient but its running time is superpolynomial. Among the deterministic algorithms running in polynomial time for directed {\it st}-reachability, the most space efficient algorithm is due to Barnes et al.~\cite{BarnesBRS98} who gave a slightly sublinear space (using $n/2^{\Theta(\sqrt{\lg n})}$ bits) algorithm for this problem
running in polynomial time. We know of no better polynomial time algorithm for this problem with better space bound. Moreover, the space used by this algorithm matches a lower bound on space for solving directed {\it st}-reachability on a restricted model of computation called Node Naming Jumping Automata on Graphs ({\sf NNJAG})~\cite{CookR80,EdmondsPA99}. This model was introduced especially for the study of directed {\it st}-reachability and most of the known sublinear space algorithms for this problem can be implemented on it. Thus, to design any polynomial time ROM algorithm taking space less than $n/2^{\Theta(\sqrt{\lg n})}$ bits requires significantly new ideas. Recently there has been some improvement in the space bound for some special classes of graphs like planar and H-minor free graphs~\cite{AsanoKNW14,ChakrabortyPTVY14}. Other than these fundamental graph theoretical problems, there have been some work on designing space-efficient algorithms for the more classical selection and sorting problems \cite{Beame91,MunroP80,MunroR96}, and problems in computational geometry~\cite{AsanoBBKMRS14,BarbaKLSS15,BarbaKLS14,DarwishE14} among others.

A drawback, however, for all these graph algorithms using small space i.e., sublinear bits, is that their running time is often some polynomial of high degree. For example, to the best of our knowledge, the exact running time of Reingold's algorithm \cite{Reingold08} for undirected {\it s-t} connectivity is not analysed, yet we know it admits a large polynomial running time. This is not surprising as Tompa~\cite{Tompa82} showed that for directed {\it st}-reachability, if the number of bits available is $o(n)$ then some natural algorithmic approaches to the problem require super-polynomial time. Motivated by these impossibility results from complexity theory and inspired by the practical applications of these fundamental graph algorithms, recently there has been a surge of interest in improving the space complexity of the fundamental graph algorithms without paying too much penalty in the running time i.e., reducing the working space of the classical graph algorithms to $O(n)$ bits with little or no penalty in running time. Generally most of the classical linear time graph algorithms take $O(n)$ words or equivalently $O(n \lg n)$ bits of space.

Starting with the paper of Asano et al. \cite{AsanoIKKOOSTU14} who showed how one can implement DFS using $O(n)$ bits, improving on the naive $O(n \lg n)$-bit implementation, the recent series of papers \cite{AsanoIKKOOSTU14,BanerjeeC016,BanerjeeCRRS2015,CJS,Chakraborty0S16,CS,ElmasryHK15} presented space-efficient algorithms for a few other basic graph problems: namely BFS, maximum cardinality search, topological sort, connected components, minimum spanning tree, shortest path, recognition of outerplanar graph and chordal graphs among others. We add to this small yet growing body of space-efficient algorithm design literature by providing such algorithms for some classical graph problems that have been solved using DFS, namely the problem of testing biconnectivity and $2$-edge connectivity, finding cut vertices and cut edges, computing chain decomposition and $st$-numbering among others.

\subsection{Model of Computation}
As is standard in the area of space-efficient graph algorithms~\cite{AsanoIKKOOSTU14,BanerjeeCRRS2015,BanerjeeC016,ElmasryHK15}, we assume that the input graph is given in a read-only memory (and so cannot be modified). If an algorithm must do some outputting, this is done on a separate write-only memory. When something is written to this memory, the information cannot be read or rewritten again. So the input is ``read only'' and the output is ``write only''. In addition to the input and the output media, a limited random-access workspace is available. The data on this workspace is manipulated at word level as in the standard word RAM model, where the machine consists of words of size $w = \Omega (\lg n)$ bits; and any logical, arithmetic, and bitwise operations involving a constant number of words take a constant amount of time. We count space in terms of the number of bits in the workspace used by the algorithms. Historically, this model is called the {\it register input model} 
and it was introduced by Frederickson \cite{Frederickson87} while studying some problems related to sorting and selection. 

We assume that the input graphs $G=(V,E)$ are represented using an {\it adjacency array}, i.e., $G$ is represented by an array of length $|V|$ where the $i$-th entry stores a pointer to an array that stores all the neighbors of the $i$-th vertex. For the directed graphs, we assume that the input representation has both in/out adjacency array for all the vertices i.e., for directed graphs, every vertex $v$ has access to two arrays, one array is for all the in-neighbors of $v$ and the other array is for all the out-neighbors of $v$. This representation which has now become somewhat standard was also used in \cite{BanerjeeC016,Chakraborty0S16,ElmasryHK15,HagerupK16} recently to design various other space efficient graph algorithms. We use $n$ and $m$ to denote the number of vertices and the number of edges respectively, in the input graph $G$. Throughout the paper, we assume that the input graph is a connected graph, and hence $m \geq n-1$. 

\subsection{Our results and organization of the paper}
Asano et al.~\cite{AsanoIKKOOSTU14} showed that Depth First Search (DFS) in a directed or an undirected graph can be performed 
in $O(m \lg n)$ time and $O(n)$ bits of space. Elmasry et al. \cite{ElmasryHK15} improved the time to $O(m\lg \lg n)$ still using $O(n)$ bits of space. We build upon these results to give space efficient implementations of several classical applications of DFS.

First, as a warm up, we start with 
some simple applications of the space efficient DFS to show the following.
\begin{itemize}
\item
An $O(m \lg n \lg \lg n)$ time and $O(n)$ bits of space algorithm to compute the strongly connected components of a directed 
graph in Section~\ref{strong_conn}. 
\end{itemize}
In addition, we also give
\begin{itemize}
\item
an algorithm to output the vertices of a directed acyclic graph in a topologically sorted order in Section~\ref{top_sort}, and
\item
an algorithm to find a sparse (with $O(n)$ edges) spanning biconnected subgraph of an undirected biconnected graph in Section~\ref{sparse_bicon_subgraph}
\end{itemize}
both using asymptotically the same time and space used for DFS, i.e., using $O(n)$ bits and $O(m \lg \lg n)$ time. 

To develop fast and space efficient algorithms for other non-trivial graph problems which are also applications of DFS, in Section~\ref{treecover}, we develop and describe in detail a space efficient tree covering technique, and use this in subsequent sections. This technique, roughly speaking, partitions the DFS tree into connected smaller sized subtrees which can be stored using less space. Finally we solve the corresponding graph problem on these smaller sized subtrees and merge the solutions across the subtrees to get an overall solution. All of these can be done using less space and not paying too much penalty in the running time. Some of these ideas are borrowed from succinct tree representation literature.

As the first application, we consider in Section~\ref{sec:chain-decomp}, a space efficient implementation of chain decomposition 
of an undirected graph. This is an important preprocessing routine for an algorithm to find cut vertices, biconnected components, cut edges, and also to test 3-connectivity~\cite{Schmidt2010c} among others. 
We provide an algorithm that takes $O(m \lg^2 n \lg \lg n)$ time using $O(n)$ bits of space, improving on previous implementations that took $\Omega (n \lg n)$ bits~\cite{Schmidt13} or $\Theta (m+n)$ bits~\cite{BanerjeeC016} of space.

In Section~\ref{sec:onbiconn}, we give improved space efficient algorithms for testing whether a given undirected graph $G$ is biconnected, and if $G$ is not biconnected, we also show how one can find all the cut vertices of $G$.
For this, we provide a space efficient implementation of Tarjan's classical lowpoint algorithm~\cite{Tarjan72}. Our algorithms take $O(m \lg n \lg \lg n)$ time and $O(n)$ bits of space. In Section~\ref{edgecon}, we provide a space efficient implementation for testing $2$-edge connectivity of a given undirected graph $G$, and producing cut edges of $G$ using $O(m \lg n \lg \lg n)$ time and $O(n)$ bits of space.

Given a biconnected graph, and two distinguished vertices $s$ and $t$, $st$-numbering is a numbering of the vertices of the graph so that $s$ gets the smallest number, $t$ gets the largest and every other vertex is adjacent both to a lower-numbered and to a higher-numbered vertex. Finding an $st$-numbering is an important preprocessing routine for a planarity testing algorithm~\cite{EvenT76} among others. In Section~\ref{sec:stnumb}, we give an algorithm to determine an $st$-numbering of a biconnected graph that takes $O(m \lg^2 n \lg \lg n)$ time using $O(n)$ bits. This improves the earlier implementations that take $\Omega (n \lg n)$ bits~\cite{Brandes02,Ebert83,EvenT76,Tarjan86}. Using this as a subroutine, in Section~\ref{st_app}, we provide improved space effcient implementation for 
two-partitioning and two independent spanning tree problem among others. We direct the readers to Section~\ref{sec:terms} where we provide all the necessary definitions.

\subsection{Related Models}
Several models of computation come close to {\it read-only random-access} model, the model we focus on this paper, when it comes to design space-efficient graph algorithms. A single thread common to all of them is that access to the input tape is restricted in some way. In the {\it multi-pass streaming} model \cite{MunroP80} the input is kept in a read-only sequentially-accessible media, and an algorithm tries to optimize on the number of passes it makes over the input. In the {\it semi-streaming} model~\cite{FeigenbaumKMSZ05}, the elements (or edges if the input is graph) are revealed one by one and extra space allowed to the algorithm is $O(n.polylg(n))$ bits. Observe that, it is not possible to store the whole graph if it is dense. The efficiency of an algorithm in this model is measured by the space it uses, the time it requires to process each edge and the number of passes it makes over the stream. In the {\it in-place} model \cite{BronnimannC06}, one is allowed a constant number of additional variables, but it is possible to rearrange (and sometimes even modify) the input values. Chan et al.~\cite{ChanMR14} introduced the~\emph{restore} model which is a more relaxed version of read-only memory (and a restricted version of the in-place model), where the input is allowed to be modified, but at the end of the computation, the input has to be restored to its original form. This has motivation, for example, in scenarios where the input (in its original form) is required by some other application. Buhrman et al.~\cite{BuhrmanCKLS14,Koucky16} introduced and studied the {\it catalytic-space} model where a small amount (typically $O(\lg n)$ bits) of clean space is provided along with additional auxiliary space, with the condition that the additional space is initially in an arbitrary, possibly incompressible, state and must be returned to this state when the computation is finished. The input is assumed to be given in ROM. They also provided implementations of some graph algorithms space efficiently.

\section{Preliminaries}\label{sec:prelims}

In this section, we list some preliminary results and graph theoretic definitions that will be used later in the algorithms we develop. We also discuss briefly, at a very high level, the main technique that goes behind almost all of our algorithms in this paper.

\subsection{Graph theoretic terminology}
\label{sec:terms}
Here we collect all the necessary graph theoretic definitions that will be used throughout the paper. A cut vertex in an undirected graph $G$ is a vertex $v$ that when removed (along with its incident edges) from a graph creates more components than previously in the graph. 
A (connected) graph with at least three vertices is biconnected (also called $2$-connected in the graph literature sometimes) if and only if it has no cut vertex. A biconnected component is a maximal biconnected subgraph. These components are attached to each other at cut vertices. Similarly in an undirected graph $G$, a bridge (or cut edge) is an edge that when removed (without removing the vertices) from a graph creates more components than previously in the graph. A (connected) graph with at least two vertices is $2$-edge-connected (also called bridgeless sometimes) if and only if it has no bridge. A $2$-edge connected component is a maximal $2$-edge connected subgraph. 

Given a biconnected graph $G$, and two distinguished vertices $s$ and $t$ in $V$ such that $s \neq t$, $st$-numbering is a numbering of the vertices of the graph so that $s$ gets the smallest number, $t$ gets the largest and every other vertex is adjacent both to a lower-numbered and to a higher-numbered vertex i.e., a numbering $s=v_1,v_2,\cdots,v_n=t$ of the 
vertices of $G$ is called an $st$-numbering, if for all vertices $v_j, 1<j<n$, there exist $1\leq i<j<k\leq n$ such that $\{v_i, v_j\},\{v_j, v_k\} \in E$. It is well-known that $G$ is biconnected if and only if, for every edge $\{s,t\}\in E$, it has an $st$-numbering. In the {\it $k$-partitioning problem}, we are given vertices $a_1,\cdots, a_k$ of an undirected graph $G$ and natural numbers $c_1,\cdots, c_k$ with $c_1+\cdots+ c_k = n$, and we want to find a partition of $V$ into sets $V_1,\cdots, V_k$ with $a_i \in V_i$ and $|V_i| = c_i$ for every $i$ such that every set $V_i$ induces a connected graph in $G$. Given a graph $G$, we call a set of $k$ rooted spanning trees independent if they all have the same root vertex $r$ and, for every vertex $v\neq r$, the paths from $v$ to $r$ in all the $k$ spanning trees are  vertex-disjoint (except for their endpoints). A directed graph $G$ is said to be {\it strongly connected} if for every pair of vertices $u$ and $v$ in $V$, both $u$ and $v$ are reachable from each 
other. If $G$ is not strongly connected, it is possible to decompose $G$ into its strongly connected components i.e. a maximal set of vertices $C \subseteq V$ such that for every pair of vertices $u$ and $v$ in $C$, both $u$ and $v$ are reachable from each other.  A topological sort or topological ordering of a directed acyclic graph is a linear ordering of its vertices such that for every directed edge $(u,v) \in E$ from vertex $u$ to vertex $v$, $u$ comes before $v$ in the ordering. Let $T$ be a depth-first search tree of a connected undirected (or directed) graph $G$. For each vertex $v$ of $T$, preorder number of $v$ is the number of vertices visited up to and including $v$ during a preorder traversal of $T$. Similarly, postorder number of $v$ is the number of vertices visited up to and including $v$ during a postorder traversal of $T$.

\subsection{Tree cover and its space efficient construction}
To implement our algorithms in $O(n)$ bits, our main idea is to process the nodes of the DFS tree in batches of $O(n/\lg n)$ nodes, as we can only afford to store trees of size $O(n/ \lg n)$ explicitly with their labels. To do this, we first use a tree-cover algorithm (that is used in succinct representations of trees) to partition the tree into $O(\lg n)$ connected subtrees of size $O(n/\lg n)$ each. We then solve the problem we are dealing with in these smaller subtrees, and later merge them in a specific order to obtain the overall solution. In some cases, to obtain the overall solution, we need to generate pairs of subtrees
with explicit node labels, and then process the edges between them in a specific order. We describe all the details of the tree cover approach in Section~\ref{treecover}, and describe the algorithms in Section~\ref{everything}. 


\subsection{Rank-Select}\label{rs}
We use the following fundamental data structure on bitstrings in some of our algorithms.
Given a bitvector $B$ of length $n$, the rank and select operations are defined as follows:
\begin{itemize}
 \item $rank_a(i,B)$ = number of occurrences of $a\in \{0,1\}$ in $B[1,i]$, for $1\leq i\leq n$;
 \item $select_a(i,B)$ = position in $B$ of the $i$-th occurrence of $a\in \{0,1\}$.
\end{itemize}
The following theorem gives an efficient structure to support these operations.
\begin{theorem}[\cite{Clark96,Munro96,MunroRR01}]
 \label{staticbit}
 Given a bitstring $B$ of length $n$, one can construct a $o(n)$-bit auxiliary structure to support rank and select operations 
in $O(1)$ time. Also, such a structure can be constructed from the given bitstring in $O(n)$ time.
\end{theorem}

\subsection{Related work on Space-efficient DFS}
Recall that DFS starts exploring the given input graph $G$ where each vertex is initially {\it white} meaning unexplored, becomes {\it gray} when DFS discovers for the first time and pushed on the stack, and is colored {\it black} when it is finished i.e. its adjacency list has been checked completely and it leaves the stack. Recently Elmasry et al.~\cite{ElmasryHK15} showed the following tradeoff result for DFS,
\begin{theorem}[\cite{ElmasryHK15}]\label{thm:elmasry-tradeoff}
For every function $t: \mathbb{N} \rightarrow \mathbb{N}$ such that $t(n)$ can be computed within the resource bound of this theorem (e.g., in $O(n)$ time using $O(n)$ bits), the vertices of a directed or undirected graph $G$ can be visited in depth first order in $O((m+n)t(n))$ time with $O(n+n\frac{\lg \lg n}{t(n)})$ bits.
\end{theorem}
In particular, fixing $t(n)=O(\lg \lg n)$, one can obtain a DFS implementation which runs in $O(m \lg \lg n)$ time using $O(n)$ bits. We build on top of this DFS algorithm to provide space efficient implementation for various applications of DFS in directed and undirected graphs in the rest of this paper.

\section{Some simple applications of DFS using $O(n)$ bits \protect\footnote{The results of this section were announced in preliminary form in the proceedings of 22nd International Computing and Combinatorics Conference (COCOON 2016), Springer LNCS volume 9797, pages 119-130~\cite{BanerjeeC016}.}}
Classical applications of DFS in directed graphs (see~\cite{CLRS}) are to find strongly connected components of a directed graph, and to do a topological sort of a directed acyclic graph among many others. Also, given an undirected biconnected graph $G$, DFS is used as the main tool to produce a sparse spanning biconnected subgraph of $G$. We show here that while topological sort and producing a sparse spanning biconnected subgraph of an undirected biconnected graph can be solved using the same $O(n)$ bits and $O(m \lg\lg n)$ time (as for DFS), strongly connected components of a directed graph can be obtained using $O(n)$ bits and $O(m \lg n \lg \lg n)$ time.

\subsection{Strongly Connected Components}\label{strong_conn}
There is a classical two pass algorithm (see \cite{CLRS} or \cite{dasgupta}) for computing the Strongly Connected Components (SCC) of a given directed graph $G$ which works as follows. In the first step, it runs a DFS on $G^R$, the reverse graph of $G$. In the second pass, it runs the connected component algorithm using DFS in $G$ but it processes the vertices in the decreasing order of the finishing time from the first pass. 

We can obtain $G^R$ by switching the role of in and out adjacency arrays present in the input representation. As we can not remember the vertex ordering from the first pass due to space restriction, we process them in batches of size $n/\lg n$ in the reverse order i.e., we run a full DFS in $G^R$ to obtain and store the last $n/\lg n$ vertices in an array $A$ as they are the ones which have the highest set of finishing numbers in decreasing order. I.e., we maintain $A$ as a queue of size $n/\lg n$ and as and when a new element is finished, it is added to the queue and the element with the earliest finish time at the other end of the queue is deleted.
Now, we pick the vertices from $A$ one by one in the order from the queue with the latest finish time and start a fresh DFS in $G$ to compute the connected components and output all the vertices reachable as a SCC. The output vertices are marked in a bitmap so that we don't output them again. Once we are done with all the vertices in $A$, we restart the DFS from the beginning and produce the next chunk of $n/\lg n$ vertices by remembering the last vertex produced in the previous step and stop as soon as we hit that boundary vertex. Then we repeat the connected component algorithm from this chunk of vertices and continue this way. It is clear that the algorithm produces the SCCs correctly. As we are calling the DFS algorithm $O(\lg n)$ times, total time taken by this algorithm is $O(m\lg\lg n\lg n)$ with $O(n)$ bits of space. Hence, we have the following,
\begin{theorem}
Given a directed graph $G$ on $n$ vertices and $m$ edges, represented as in/out adjacency array, we can output the strongly connected components of $G$ in $O(m \lg n \lg \lg n)$ time and $O(n)$ bits of space.
\end{theorem}

\subsection{Topological Sort}\label{top_sort}
The standard algorithm for computing topological sort~\cite{CLRS} outputs the vertices of a DFS in reverse order.
If we can keep track of the DFS numbers, then reversing is an easy task. 
While working in space restricted setting (with $o(n \lg n)$ bits), this is a challenge as we don't have space to keep track of 
the DFS order. We can do as we did in the strongly connected components algorithm in the last section, by storing and outputting vertices in batches of $n/\lg n$ resulting in an $O(m \lg n \lg \lg n)$ time algorithm.

Elmasry et al.\cite{ElmasryHK15} showed that, the vertices of a DAG $G$ can be output in the order of a topological sort within the time and space bounds of a DFS in $G$ plus an additional $O(n \lg \lg n)$ bits. As they also showed how to perform DFS in $O(m+n)$ time and $O(n \lg \lg n)$ bits, overall their algorithm takes $O(m+n)$ time and $O(n \lg \lg n)$ bits to compute a topological sorting of $G$. Their main idea is to maintain enough information about a DFS to resume it in the middle and apply this 
repeatedly to reverse small chunks of its output, produced in reverse order, one by one.

We observe that, instead of storing information to restart DFS and produce the reverse order, we simply work with the reverse graph itself (which can be obtained from the input representation by switching the role of in and out adjacency arrays) and do a DFS in the reverse graph and output vertices as they are finished (or blackened) i.e., in the increasing order of finishing time. 
To see the correctness of this procedure, note that the reverse graph is also a DAG, and if $(i,j)$ is an edge in the DAG $G$, then $(j,i)$ is an edge in the reverse graph and $i$ will become black before $j$ while the algorithm performs DFS in the reverse graph. Hence, $i$ will be placed before $j$ in the correct topological sorted order. Thus we have the following, 
\begin{theorem}
\label{dfstopo}
Given a DAG $G$ on $n$ vertices and $m$ edges, if the black vertices of the DFS of $G$ can be output using 
$s(n)$ space and $t(n)$ time, then its vertices can be output in topologically sorted order using $O(s(n))$ space and $O(t(n))$ time 
assuming that the input representation has both the in and out adjacency array of the graph.
\end{theorem}
From Theorem~\ref{thm:elmasry-tradeoff} (setting $t(n)=O(\lg \lg n)$) and Theorem~\ref{dfstopo}, we have the following.
\begin{corollary}
\label{topo}
Given a DAG $G$ on $n$ vertices and $m$ edges, its vertices can be output in topologically sorted order using 
$O(m\lg \lg n)$ time and $O(n)$ bits.
\end{corollary}

Note that,
we knew all along that DFS and topological sort take the same time, the main contribution of Theorem~\ref{dfstopo} is that it shows they take the same space (improving on the result of~\cite{ElmasryHK15} where they showed that topological sort space = DFS space + $O(n \lg \lg n)$ bits under the same time) when both the in/out adjacency arrays are present in the input.

\subsubsection{Topological Sort in Sublinear Space}
We note the following theorem of Asano et al.~\cite{AsanoIKKOOSTU14}.
\begin{theorem}
DFS on a DAG $G$ can be performed in space $O(\frac{n}{2^{(\sqrt{\lg n})}})$ bits and in polynomial time. 
\end{theorem}

While it should immediately follow from Theorem~\ref{dfstopo} that topological sort can also be performed using such sublinear bits of space, there is one caveat. Asano et al.'s algorithm works assuming that the given DAG $G$ has a single source vertex. In particular, they determine whether a vertex is black by checking whether it is reachable from {\it the} source without using the gray vertices (using the sublinear space reachability algorithm of~\cite{BarnesBRS98}).

The algorithm can be easily extended to handle $s$ many sources if we have some additional $s \log n$ bits. We 
simply keep track of the indices of the sources from which DFS has been explored, and to determine whether a vertex is black, we ask if it is reachable from an earlier source or from the current source without using the gray vertices. 
Thus we have the following improved theorem.
\begin{theorem}\label{modified_dfs}
DFS on DAG $G$ with $s$ sources can be performed using $s \lg n + o(n)$ bits and polynomial time. In particular, if $s$ is $o(n/\lg n)$, the overall space used is $o(n)$ bits.
\end{theorem}
Thus from Theorem \ref{dfstopo} and Theorem~\ref{modified_dfs} we obtain the following,
\begin{theorem}
\label{sub_topo}
Topological Sort on a DAG $G$ with $s$ sinks can be performed using $s \lg n + o(n)$ bits and polynomial time. In particular 
if $s$ is $o(n/\lg n)$, the overall space used is $o(n)$ bits.
\end{theorem}

\subsection{Finding a sparse biconnected subgraph of a biconnected graph}\label{sparse_bicon_subgraph}
The problem of finding a $k$-connected spanning subgraph with the minimum number of edges of a $k$-connected graph is known to be NP-hard for any $k\geq2$~\cite{GareyJ79}. But the complexity of the problem decreases drastically if all we want is to produce a ``sparse'' $k$-connected spanning subgraph, i.e., one with $O(n)$ edges. Nagamochi and Ibaraki~\cite{NagamochiI92} gave a linear time algorithm which produces a $k$-connected spanning subgraph with at most $kn-\frac{k(k+1)}{2}$ edges. Later, Cheriyan et al.~\cite{CheriyanKT93} gave another linear time algorithm for $k=2$ and $3$ that produced a $2$-connected spanning subgraph with at most $2n-2$ edges, and a $3$-connected subgraph with at most $3n-3$ edges. Later, Elmasry~\cite{Elmasry10} gave an alternate linear time algorithm for producing a sparse spanning biconnected subgraph of a given biconnected graph by performing a DFS with additional bookkeeping. In what follows, we provide a space efficient implementation for it. In order to do that, we start by briefly 
describing Elmasry's algorithm.

Let $DFI(v)$ denote the index (integer) that represents the time at which the vertex $v$ is first discovered from the vertex $u$ when performing a DFS i.e., $u$ is the parent of $v$ in the DFS tree. Let $low(v)$ be the smallest $DFI$ value among the $DFI$ values of vertices $w$ such that $(v,w)$ is a back edge. (Note that this quantity is different from the ``lowpoint'' value used in Tarjan's \cite{Tarjan72} classical biconnectivity algorithm.) Basically $low(v)$ captures the information regarding the deepest back edge going out of the vertex $v$. If $v$ has no backedges, for convenience (the reason will become clear in the following lemma), we adopt the convention that $low(v)=DFI(parent(v))$. The edge $(v,low(v))$ is the deepest backedge out of $v$. Note that, it is actually the tree edge between $v$ and its parent if $v$ does not have a backedge. The algorithm maintains all the edges of the DFS tree. In addition, for every vertex in the graph, the algorithm maintains the $DFI$ and the $low$ values along with the back edge that realizes it. As the root of the DFS tree does not have any back edge and, as the underlying graph is 2-connected, the root has only one child $v$ so that there is no back edge emanating from $v$ as well. Thus we get at most $n-2$ back edges along with $n-1$ tree edges, giving a subgraph with at most $2n-3$ edges. Elmasry~\cite{Elmasry10} proved that the resulting graph is indeed a spanning 2-connected subgraph of $G$. His algorithm takes $O(m+n)$ time and $O(n \lg n)$ bits of space. We improve the space bound, albeit with slight increase in time, by first proving a more general lemma as following,

\begin{figure}[h]
\begin{center}
 \includegraphics[scale=.8, keepaspectratio=true]{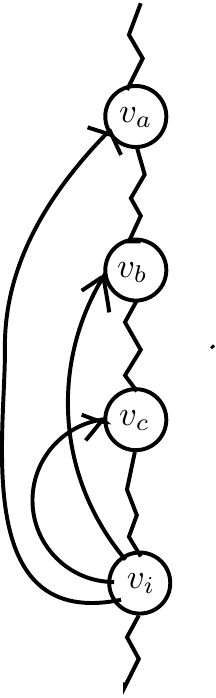}
\end{center}
\caption{A part of the full DFS tree. The wiggling edges represent tree edges and the edges with arrow heads represent back edges. If $low(v_i)=v_a$, we would come across $v_i$ in the adjacency array of $v_a$ before encountering from the arrays of $v_b$ and $v_c$. I.e., the back edge $(v_a, v_i)$ will be processed before the other back edges $(v_b, v_i)$ and $(v_c, v_i)$ since we process the vertices (and the backedges incident to them) in their DFS order.}
\end{figure}

\begin{lemma}\label{dbe_lemma}
Given any undirected graph $G$ with $n$ vertices and $m$ edges, we can compute and report the $low(v)$ values i.e., deepest back edge going out of $v$, for every vertex $v$, using $O(n)$ bits of space and $O(m \lg \lg n)$ time. 
\end{lemma}

\begin{proof}
The aim is to output all the deepest back edges out of every vertex $v$ in $G$ as we perform the DFS. As always, let $\{v_1, v_2,\cdots, v_n\}$ be the vertices of the graph. We perform a DFS with the usual color array and other relevant data structures (as required in Theorem~\ref{thm:elmasry-tradeoff} with $t(n)=\lg \lg n$) along with one more array of $n$ bits, which we call $DBE$ (for Deepest Back Edge) array, which is initialized to all zero. $DBE[i]$ is set to 1 if and only if the algorithm has found and output the deepest back edge emanating from vertex $v_i$. So whenever a white vertex $v_i$ becomes gray (i.e., $v_i$ is visited for the first time), we scan $v_i$'s adjacency array to mark, for every white neighbor $v_j$, $DBE[j]$ to $1$ if and only if it was $0$ before. The correctness of this step follows from the fact that as we are visiting vertices in DFS order, and if $DBE[j]$ is $0$, then vertex $v_j$ is not adjacent to any of the vertices we have visited so far, and as it is adjacent to $v_i$, the deepest back edge emanating from $v_j$ is $(v_i, v_j)$. Hence we output this edge and move on to the next neighbor and eventually with the next step of DFS until all the vertices are exhausted. This completes the description of the algorithm. See Figure $1$ for an illustration. Now to see how this procedure produces all the deepest back edges out of every vertex, note that, at vertex $v_i$, our algorithm reports all the back edges $e = (v_i, v_j)$ where $e$ is the deepest back edge from $v_j$, and also all tree edges $(v_i, v_j)$ where $v_j$ has no back edge. Observe that from our convention, in the second case, $(v_i, v_j)$ is the deepest back edge out of $v_j$. This concludes the proof of the lemma. As we performed just one DFS to produce all such edges, using Theorem~\ref{thm:elmasry-tradeoff}, the claimed running time and space bounds follow. 
\end{proof}

The way we will actually use Lemma~\ref{dbe_lemma} in our algorithms, is for finding and storing the $low$ values for at most $n/\lg n$ vertices. So we state a corollary for that.
\begin{corollary}\label{coro}
Given any undirected graph $G$ with $n$ vertices and $m$ edges and any set $L$ of $O(n/\lg n)$ vertices as input, we can compute, report and store the $low(v)$ values for every vertex $v$ in $L$ in the DFS tree $T$ of $G$ using $O(n)$ bits of space and $O(m \lg \lg n)$ time.
\end{corollary}

Note that, Lemma~\ref{dbe_lemma} holds true for any undirected connected graph $G$. In what follows, we use Lemma~\ref{dbe_lemma} to give a space efficient implementation of Elmasry's algorithm when the input graph $G$ is an undirected biconnected graph. In particular, we show the following, 

\begin{theorem}
\label{sparse}
Given an undirected biconnected graph $G$ with $n$ vertices and $m$ edges, we can output the edges of a sparse spanning biconnected subgraph of $G$ using $O(n)$ bits of space and $O(m \lg \lg n)$ time.
\end{theorem}

\begin{proof}
When the underlying graph $G$ is undirected biconnected graph, we know that Elmasry's algorithm produces a sparse spanning subgraph which is also biconnected. In order to implement that, given an undirected biconnected graph $G$, we first run on $G$ the algorithm of Lemma~\ref{dbe_lemma} which produces and reports all the deepest back edges out of all the vertices $v$ in $G$. Out of all those deepest back edges, note that, some are actually tree edges from our convention. Hence, we don't want to report them multiple time. More specifically, if a vertex $v_j$ has no back edge going out of it, Lemma~\ref{dbe_lemma} outputs the edge $(v_i, v_j)$ as the deepest back edge out of $v_j$, which is actually a tree edge in the DFS tree $T$ of $G$. In order to avoid reporting such edges more than once, we perform the following. During the scanning of $v_i$'s adjacency array, we also check if any of its neighbor, other than its parent, is gray. If so, we report the edge from $v_i$ to its parent. Note that if $v_i$ has a back edge to one of its ancestors (other than its parent), then this step reports the tree edge from $v_i$ to its parent. Otherwise, $v_i$ didn't have any back edge, and hence the tree edge to its parent would have been output while DFS was exploring and outputting deepest back edges from its parent; so we do not output the edge again. Note that, we can do this test along with the algorithm of Lemma~\ref{dbe_lemma} so that using just one DFS, we can produce all the tree edges and deepest back edges as required in Elmasry's algorithm. Thus using Theorem~\ref{thm:elmasry-tradeoff}, we can output the edges of a sparse spanning biconnected subgraph of $G$ using $O(n)$ bits of space and $O(m \lg \lg n)$ time.
\end{proof}

\section{Tree Cover and Space Efficient Construction}
\label{treecover}
Before moving on to handle other complex applications of DFS in undirected graphs, namely biconnectivity, $2$-edge connectivity, $st$-numbering etc, in the this section we discuss the common methodology to attack all of these problems. Once we set all our machinary here, in Section \ref{treecover}, we see afterwards how to use them almost in a similar fashion to several problems. Central to all of our algorithms following this section is a decomposition of the DFS tree. For this we use the well-known 
tree covering technique which was first proposed by Geary et al. \cite{GearyRR06} in the context 
of succinct representation of rooted ordered trees. 
The high level idea is to decompose the tree into subtrees called 
{\it minitrees}, and further decompose the mini-trees into yet smaller subtrees called 
{\it microtrees}. The microtrees are tiny enough to be stored in a compact table. The root of a 
minitree can be shared by several other minitrees. To represent 
the tree, we only have to represent the connections and links between the subtrees. 
Later He et al.~\cite{HeMS12} extended this approach to produce a representation which 
supports several additional operations. Farzan and Munro~\cite{FarzanM11} modified 
the tree covering algorithm of~\cite{GearyRR06} so that each minitree has at most one 
node, other than the root of the minitree, that is connected to the 
root of another minitree. This simplifies the representation of the tree, and 
guarantees that in each minitree, there exists at most one non-root node which is 
connected to (the root of) another minitree. 
The tree decomposition method of Farzan and Munro~\cite{FarzanM11} 
is summarized in the following theorem:

\begin{figure}[h]
\begin{center}
 \includegraphics[scale=.6, keepaspectratio=true]{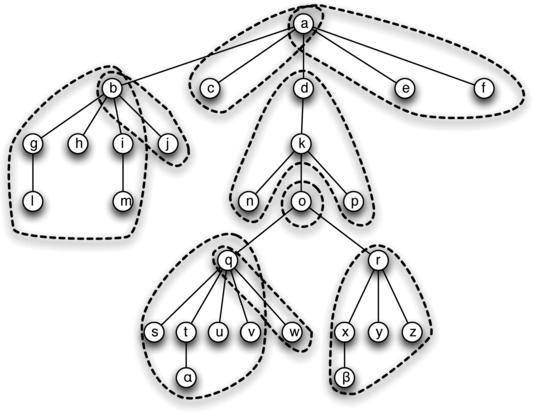}
\end{center}
\caption{An illustration of Tree Covering technique with $L=5$. The figure is reproduced from~\cite{FarzanM11}. Each closed region formed by the dotted lines represents a minitree. Note that each minitree has at most one `child' minitree (other than the minitrees that share its root) in this structure.}
\end{figure}

\begin{theorem}[\cite{FarzanM11}]\label{thm:tree-decomposition}
For any parameter $L \ge 1$, a rooted ordered tree with $n$ nodes can be decomposed into $\Theta(n/L)$ minitrees of size at most $2L$ which are pairwise disjoint aside from the minitree 
roots. Furthermore, aside from edges stemming from the minitree root, there is at most one edge 
leaving a node of a minitree to its child in another minitree. The decomposition can be performed 
in linear time.
\end{theorem}

See Figure $2$ for an illustration. In our algorithms, we apply Theorem~\ref{thm:tree-decomposition} with $L =n/\lg n$. 
For this parameter $L$, since the number of minitrees is only $O(\lg n)$, we can represent the
structure of the minitrees within the original tree (i.e., how the minitrees are connected with each other) 
using $O(\lg^2 n)$ bits. The decomposition algorithm of~\cite{FarzanM11} ensures that each 
minitree has at most one `child' minitree (other than the minitrees that share its root) in this structure. We use this property crucially in our algorithms.
We refer to this as the {\it minitree-structure}. See Figure~$3(a)$ for the minitree structure of the tree decomposition shown in Figure~$2$.

\begin{figure}[h]
\begin{center}
 \includegraphics[scale=.8, keepaspectratio=true]{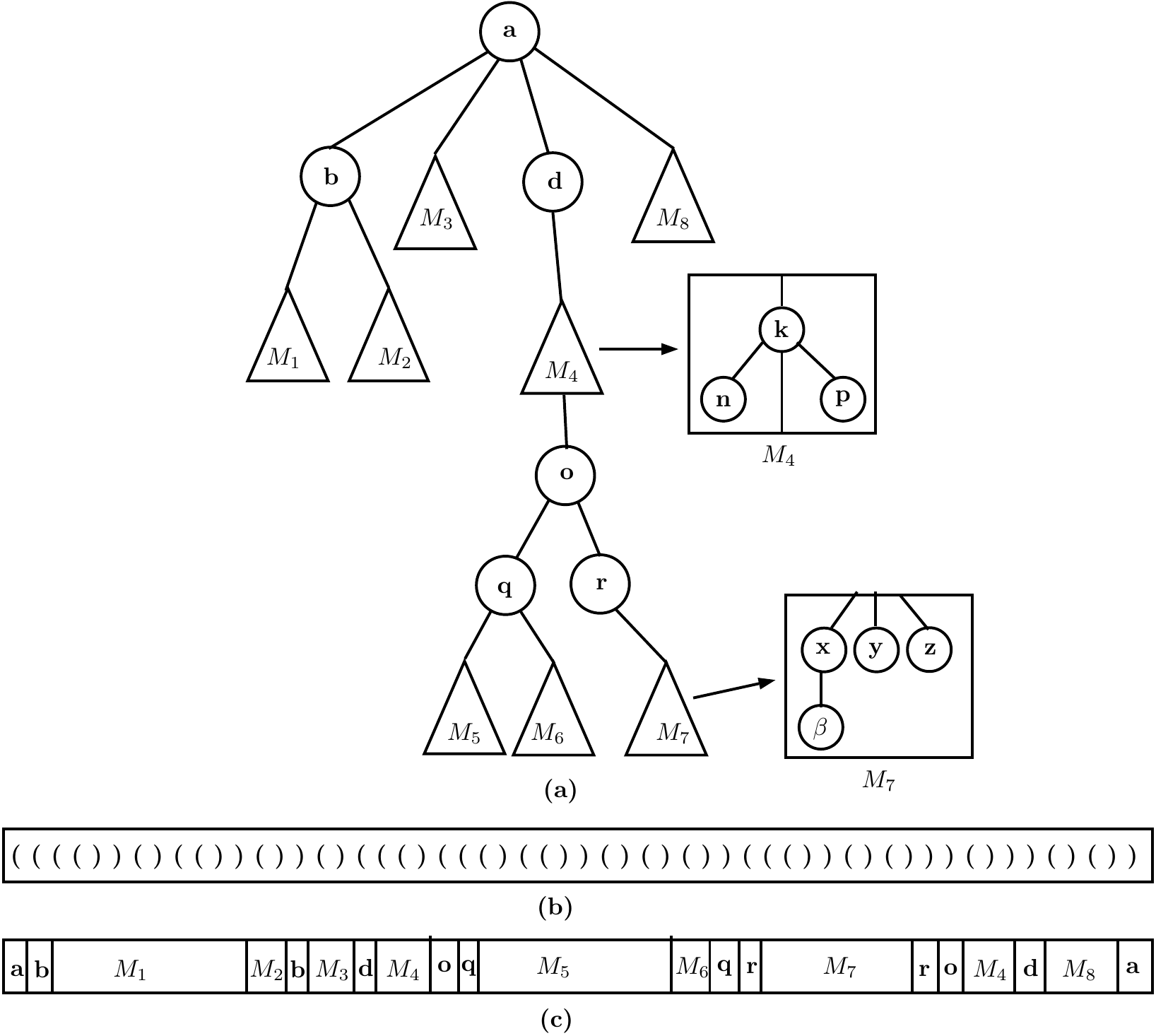}
\end{center}
\caption{ (a) The minitree structure of the tree decomposition shown in Figure~2. (b) This array encodes the entire DFS tree using the balanced parenthesis (BP) representation. (c) In this array, we demonstrate how the minitrees are split into a constant number of consecutive chunks in the BP representation. Note that the bottom array can actually be encoded using $O(\lg^2 n)$ bits, by storing, for each of the $O(\lg n)$ minitrees, pointers to all the chunks in BP sequence indicating the starting and ending positions of the chunks corresponding to the minitrees.}
\end{figure}


Explicitly storing all the minitrees (using pointers) requires $\omega(n)$ bits overall.
One way to represent them efficiently using $O(n)$ bits is to store each of them using any linear-bit
encoding of a tree~\cite{Raman013}. But if we store these minitrees separately, we loose the ability to compute
the preorder or postorder numbers of the nodes in the entire tree, which is needed in our algorithms. 
Hence, we encode the entire tree structure using a linear-bit encoding, and store pointers into this encoding to represent the minitrees.
We first encode the tree using the {\em balanced parenthesis} (BP) representation~\cite{Lu,MunroR01},
summarized in the following theorem.\footnote{The representation of ~\cite{MunroR01} does not support computing the $i$-th child of a node in constant time while the one in ~\cite{Lu} can. When using these representations to produce a tree cover, the representation of ~\cite{MunroR01} is sufficient as we just need to compute the `next child' as we traverse the tree in post-order computing the
subtree sizes of each subtree.}

\begin{theorem}[\cite{Lu}]\label{thm:BP}
Given a rooted ordered tree $T$ on $n$ nodes, it can be represented as a sequence of balanced 
parentheses of length $2n$. Given the preorder or postorder number of a node $v$ in $T$, we can support subtree size and various navigational queries (such as parent and $i$-th child) on $v$ in $O(1)$ time 
using an additional $o(n)$ bits.
\end{theorem}

The following lemma by Farzan et al.~\cite[Lemma 2]{FarzanRR09} (restated) shows 
that each minitree is split into a constant number of consecutive chunks in the BP sequence. 
So we now represent each minitree by storing pointers to the set of all {\em chunks} 
in the BP representation that together constitute the minitree. 

\begin{lemma}
In the BP sequence of a tree, the bits corresponding to a mini-tree form a set of 
constant number of substrings. Furthermore, these substrings concatenated 
together in order, form the BP sequence of the mini-tree. 
\end{lemma}
%
Hence, one can store a representation of the minitrees by storing 
an $O(\lg^2 n)$-bit structure that stores pointers to the starting positions of the chunks 
corresponding to each minitree in the BP sequence
%
We refer to the representation obtained using this tree covering (TC) approach as the TC 
representation of the tree. See Figure $3$ for a complete example of a minitree structure along with the BP sequence of the tree of Figure $2$.
The following lemma 
shows that we can construct the TC representation of the DFS tree of a given graph, using $O(n)$ additional bits.

\begin{lemma}\label{lem:BPtoTC}
Given a graph $G$ on $n$ vertices and $m$ edges, if there is an algorithm that takes $t(n,m)$ time
and $s(n,m)$ bits to perform DFS on $G$, then one can create the TC representation of the DFS 
tree in $t(n,m)+O(n)$ time, using $s(n,m)+O(n)$ bits.
\end{lemma}

\begin{proof}
%
%
We first construct the balanced parenthesis (BP) representation of the DFS tree as follows. 
We start with an empty sequence, BP, and append parentheses to it as we perform each step 
of the DFS algorithm. In particular, whenever the DFS visits a vertex $v$ for the first time, 
we append an open parenthesis to BP. Similarly when DFS backtracks from $v$, 
we append a closing parenthesis. At the end of the DFS algorithm, as every vertex is 
assigned a pair of parenthesis, length of BP is $2n$ bits. We just need to run the DFS 
algorithm once to construct this array, hence the running time of this algorithm is 
asymptotically the same as the running time of the DFS algorithm. 

We construct auxiliary structures to support various navigational operations on 
the DFS tree using the BP sequence, as mentioned in Theorem~\ref{thm:BP}. This takes $o(n)$ time and space using the algorithm of~\cite{GearyRRR06}. We then use the BP sequence along with the auxiliary structures to navigate the 
DFS tree in postorder, and simulate the tree decomposition algorithm of Farzan and 
Munro~\cite{FarzanM11} for constructing the TC representation of the DFS tree. 
If we reconstruct the entire tree (with pointers), then the intermediate space would 
be $\Omega(n \lg n)$ bits. Instead, we observe that the tree decomposition 
algorithm of~\cite{FarzanM11} never needs to keep more than $O(L)$ 
{\em temporary components} (see~\cite{FarzanM11} for details) in addition to 
some of the {\em permanent components}. Each component (permanent or temporary)
can be stored by storing the root of the component together with its subtree size.
Since $L = n/\lg n$, and the number of permanent components is only $O(\lg n)$, 
the space required to store all the permanent and temporary components at any point 
of time is bounded by $O(n)$ bits. The construction algorithm takes $O(n)$ time. 
\end{proof}

We use the following lemma 
in the description of our algorithms in the later sections.

\begin{lemma}\label{lem:DFS-minitree}
Let $G$ be a graph, and $T$ be its DFS tree. If there is an algorithm that takes $t(n,m)$ time
and $s(n,m)$ bits to perform DFS on $G$, then, using $s(n,m)+O(n)$ bits, one can 
reconstruct any minitree given by its ranges in the BP sequence of the TC representation of $T$, along with the labels of the 
corresponding nodes in the graph in $O(t(n,m))$ time.
\end{lemma}

\begin{proof}
We first perform DFS to construct the BP representation of the DFS tree, $T$. We then 
construct the TC representation of $T$, as described in Lemma~\ref{lem:BPtoTC}.
We now perform DFS algorithm again, keeping track of the preorder number of the 
current node at each step. Whenever we visit a new node, we check its preorder 
number to see if it falls within the ranges of the minitree that we want to reconstruct. 
(Note that, as mentioned above, from~\cite[Lemma 2]{FarzanRR09}, the set of 
all preorder number of the nodes that belong to any minitree form a constant number 
of ranges, since these nodes belong to a constant number of chunks in the BP 
sequence.) If it is within one of the ranges corresponding to the minitree being 
constructed, then we add the node along with its label to the minitree. 
\end{proof}

\section{Applications of DFS using tree-covering technique}\label{everything}
In this section, we provide $O(n)$ bit implementations of various algorithmic graph problems that use DFS, by using the tree covering technique developed in the previous section. At a higher level, we use the tree covering technique to generate the minitrees one by one, and then partially solve the corresponding graph problem inside that minitree before finally combining the solution across all the minitrees. 
The problems we consider include algorithms to test biconnectivity, $2$-edge connectivity and to output cut vertices, edges, and to find a chain decomposition and an $st$-numbering among others.
To test for biconnectivity and related problems, the classical algorithm due to Tarjan~\cite{Tarjan72,Tarjan74} computes the so-called ``low-point" values (which are defined in terms of a DFS-tree) for every vertex $v$, and checks some conditions based on these values. 
Brandes~\cite{Brandes02} and Gabow~\cite{Gabow00} gave considerably simpler algorithms for testing biconnectivity and computing biconnected components by using some path-generating rules instead of low-points; they call these algorithms path-based. An algorithm due to Schmidt \cite{Schmidt13} is based on chain decomposition of graphs to determine biconnectivity (and/or $2$-edge connected). All these algorithms take $O(m+n)$ time and $O(n)$ words of space. Roughly these approaches compute DFS and process the DFS tree in specific order maintaining some auxiliary information of the nodes. We start with a brief description of chain decomposition and its application first before providing its space efficient implementation. 

\subsection{Chain decomposition}\label{sec:chain-decomp}
Schmidt~\cite{Schmidt2010c} introduced a decomposition of the input graph that partitions the edge set of the graph into cycles and paths, called chains, and used this to design an algorithm to find cut vertices and biconnected components \cite{Schmidt13} and also to test 3-connectivity~\cite{Schmidt2010c} among others. In what follows, we discuss briefly the decomposition algorithm, and state his main result.

The algorithm first performs a depth first search on $G$. Let $r$ be the root of the DFS tree $T$ of $G$. DFS assigns an index to every vertex $v$, namely, the time vertex $v$ is discovered for the first time during DFS -- call it the depth-first-index of $v$ ($DFI(v)$). Imagine that the back edges are directed away from $r$ and the tree edges are directed towards $r$. The algorithm decomposes the graph into a set of paths and cycles called chains as follows. See Figure $4$
for an example. First we mark all the vertices as unvisited. Then we visit every vertex starting at $r$ in the increasing order of DFI, and do the following. For every back edge $e$ that originates at $v$, we traverse a directed cycle or a path. This begins with $v$ and the back edge $e$ and proceeds along the tree towards the root and stops at the first visited vertex or the root. During this step, we mark every encountered vertex as visited. This forms the first chain. Then we proceed with the next back edge at $v$, 
if any, or move towards the next vertex in the increasing DFI order and continue the process. Let $D$ be the collection of all such cycles and paths. Notice that the cardinality of this set is exactly the same as the number of back edges in the DFS tree as each back edge contributes to a cycle or a path. Also, as initially every vertex is unvisited, the first chain would be a cycle as it would end in the starting vertex. Using this, Schmidt proved the following theorem.

\begin{theorem}[\cite{Schmidt13}]\label{2ec}
Let $D$ be a chain decomposition of a connected graph $G(V,E)$. Then $G$
is 2-edge-connected if and only if the chains in $D$ partition $E$. Also, $G$
is 2-vertex-connected if and only if $\delta(G) \geq 2$ (where $\delta(G)$
denotes the minimum degree of $G$) and $D_1$ is the only cycle in the set
$D$ where $D_1$ is the first chain in the decomposition. An edge $e$ in
$G$ is bridge if and only if $e$ is not contained in any chain in $D$. A vertex
$v$ in $G$ is a cut vertex if and only if $v$ is the first vertex of a cycle in
$ D \setminus D_1$.
\end{theorem}

\begin{figure}[h]
\begin{center}
 \includegraphics[scale=.6, keepaspectratio=true]{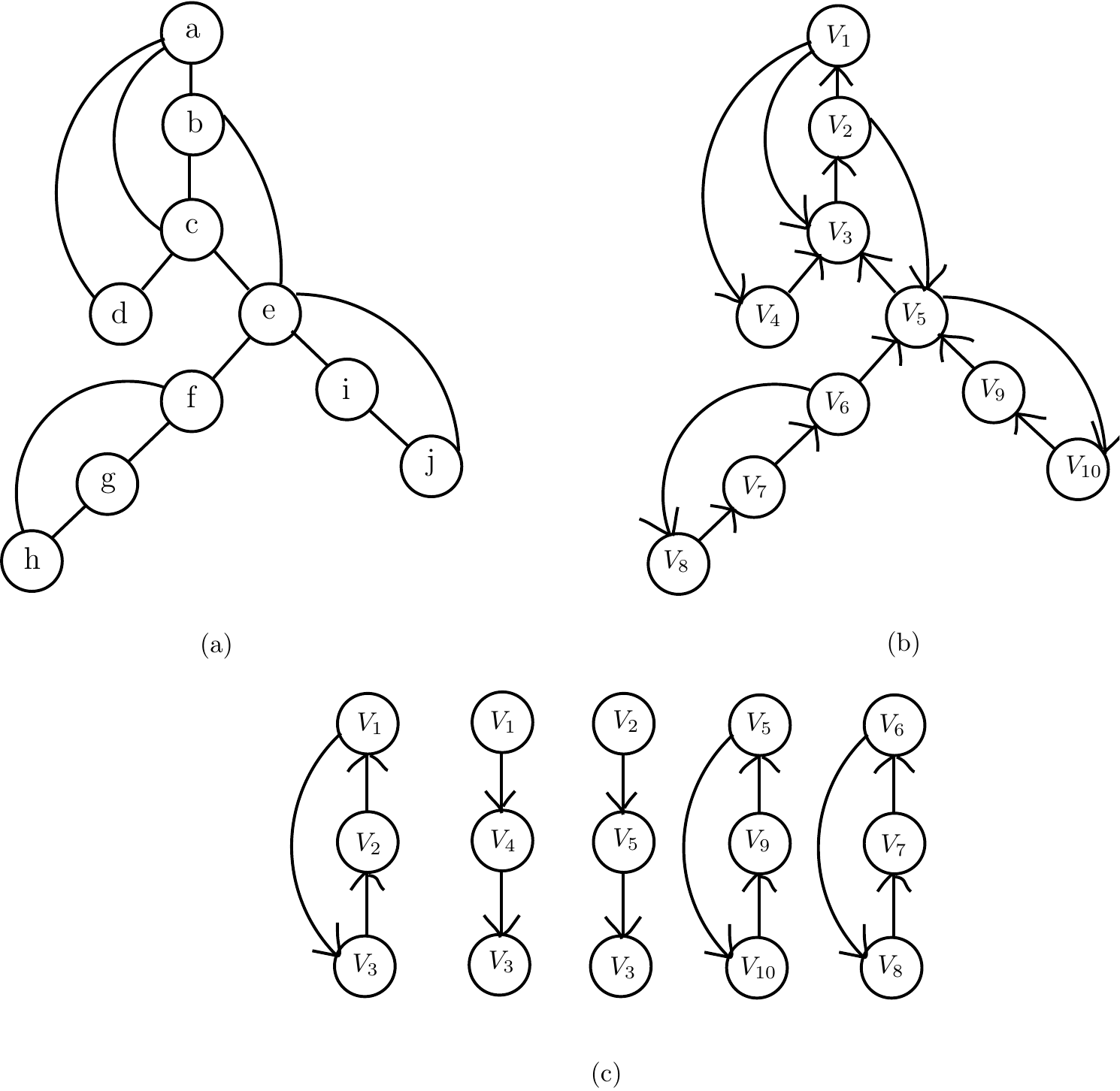}
\end{center}
\caption{Illustration of Chain Decomposition. (a) An input graph $G$. (b) A DFS traversal of $G$ and the resulting edge-orientation along with DFIs. (c) A chain decomposition $D$ of $G$. The chains $D_2$ and $D_3$ are paths and rest of them are cycles. The edge $(V_5,V_6)$ is bridge as it is not contained in any chain. $V_5$ and $V_6$ are cut vertices.}
\end{figure}
Now we are ready to describe an implementation of Schmidt's chain decomposition algorithm using only $O(n)$ bits of space and in $O(m \lg^2 n \lg\lg n)$ time using our partition of the DFS tree of Section~\ref{treecover}. 
%
%
In the following description, {\em processing a back edge} refers to the step of outputting the chain (directed path or cycle) containing that edge and marking all the encountered vertices as visited.
Processing a node refers to processing all the back edges out of that node.
The main idea of our implementation is to
process all the back edges out of each node in their {\em preorder} (as in Schmidt's algorithm).
To perform this efficiently (within the space limit of $O(n)$ bits), we process the nodes in {\em chunks} of size $n/\lg n$ each (i.e., the first chunk of $n/\lg n$ nodes in preorder are processed, 
followed by the next chunk of $n/\lg n$ nodes, and so on).
But when processing the back edges out of a chunk $C$, we process all the back edges that go from
$C$ to all the minitrees in their {\em postorder}, processing all the edges from $C$ to a minitree $\tau_1$ 
before processing any other back edges going out of $C$ to a different minitree. This requires us to go 
through all the edges out of each chunk at most $O(\lg n)$ times (once for each minitree). 
Thus the order in which we process the back edges is different from the order in which we process them 
in Schmidt's algorithm, but we argue that this does not affect the correctness of the algorithm.
In particular, we observe the following: 
\begin{itemize}
\item Schmidt's algorithm correctly produces a chain decomposition even if we process vertices to any order, as long as we process a vertex $v$ only after all its ancestors are also processed -- for example, in level order instead of preorder. This also implies that as long as we process the back edges coming to a vertex $v$ (from any of its descendants) only after we process all the back edges going to any of it's ancestors from any of $v$'s descendants, we can produce a chain decomposition correctly.
\end{itemize}

To process a back edge $(u,v)$ between a chunk $C$ and a minitree
$\tau$, where $u$ belongs to $C$, $v$ belongs to $\tau$, and $u$ is 
an anscestor of $v$, we first output the edge $(u,v)$, and then traverse the
path from $v$ to the root of $\tau$, outputting all the traversed edges and 
marking the nodes as visited. We then start another DFS to 
produce the minitree $\tau_p$ containing the parent $p$ of the root of $\tau$, 
and output the path from 
$p$ to the root of $\tau_p$, and continue the process untill we reach a vertex that has been
marked as visited. Note that this process will terminate since $u$ is marked and is an ancestor of $v$.
We maintain a bitvector of length $n$ to keep track of the marked vertices, to perform this efficiently.
A crucial observation that we use in bounding the runtime is that once we produce a minitree $\tau_p$
for a particular pair $(C,\tau)$, we don't need to produce it again, as the root of
$\tau$ will be marked after the first time we output it as part of a chain.
Also, once we generate a chunk $C$ and a minitree $\tau$, we go through all the vertices of 
$C$ in preorder, and process all the edges that go between $C$ and $\tau$.
We provide the pseudocode (see Algorithm $1$) below
describing the high-level algorithm for outputting the chain decomposition. 

\begin{algorithm}[h]
\label{general}
  \begin{algorithmic}[1]
 \Statex{Let $\tau_1,\tau_2,\cdots,\tau_{O(\lg n)}$ be the minitrees in postorder and $C_1,C_2,\cdots,C_{\lg n}$ be the chunks of vertices in preorder}   
   \For{$i = 1$ to $\lg n$} 
      \For{$j = 1$ to $O(\lg n)$} 
          \ForAll{back edges $(u,v)$ with $u \in C_i$ and $v \in \tau_j$} 
             \State{output the chain containing the edge $(u,v)$} 
          \EndFor
      \EndFor
  \EndFor
\end{algorithmic}
   \caption{Chain Decomposition}
\end{algorithm} 

The time taken for the initial part, where we construct the DFS tree, decompose it into minitrees, and construct the auxiliary structures, is $O(m \lg\lg n)$, using Theorem~\ref{thm:elmasry-tradeoff} with $t(n) = \lg\lg n$. The running time of the rest of the algorithm is dominated by the cost of processing the back edges. 
As outlined in Algorithm $1$, we process the back edges between every pair $(C_i, \tau_j)$, where 
$C_i$ is the $i$-th chunk of $n/\lg n$ nodes in preorder, and $\tau_j$ is the $j$-th minitree in postorder, for $1 \le i \le \lg n$ and $1 \le j \le O(\lg n)$. The outer loop of the algorithm generates each chunk in preorder, and thus requires a signle DFS to produce all the chunks over the entire execution of the algorithm. The inner loop goes through all the minitrees for each chunk. Since there are $\lg n$ chunks and $O(\lg n)$ minitrees, and prodicing each minitree takes $O(m \lg\lg n)$ time, the generation of all the chunk-minitree pairs takes $O(m \lg\lg n \lg^2 n)$ time. 

For a particular pair $(C, \tau)$, we may need to generate many ($O(\lg n)$, in the worst-case) minitrees. But we observe that, this happens for at most one back edge for a every pair $(C, \tau)$, since after processing the first such back edge, the root of the minitree $\tau$ is marked and hence any chain that is output afterwards will stop before the root of the minitree.
Also, if a minitree $\tau_\ell$ is generated when processing a pair $(C,\tau)$, then it will not be generated when processing any other pair $(C', \tau')$ which is different from $(C, \tau)$ (since each minitree has at most one child minitree).
Thus the overall running time is dominated by generating all the pairs $C, \tau)$, which is $O(m \lg^2 n \lg\lg n)$.
Thus, we obtain the following.

\begin{theorem}\label{thm:chain-decomp}
Given an undirected graph $G$ on $n$ vertices and $m$ edges, we can output a chain decomposition of $G$ 
in $O(m \lg^2 n \lg \lg n)$ time using $O(n)$ bits. 
\end{theorem}

\subsection{Testing biconnectivity and finding cut vertices}
\label{sec:onbiconn}

A na\"{\i}ve algorithm to test for biconnectivity of a graph $G=(V,E)$ is to check if $(V \setminus \{v\},E)$ is connected, for each $v \in V$. Using the $O(n)$ bits and $O(m+n)$ time BFS algorithm \cite{BanerjeeC016} for checking connectivity, this gives a simple $O(n)$ bits algorithm running in time $O(mn)$. Another approach is to use Theorem \ref{thm:chain-decomp} combining with criteria mentioned in Theorem \ref{2ec} to test for biconnectivity and output cut vertices in $O(m \lg^2 n \lg \lg n)$ time using $O(n)$ bits. 

Here we show that using $O(n)$ bits we can design an even faster algorithm running in $O(m \lg n \lg \lg n)$ time. If $G$ is not biconnected, then we also show how one can find all the cut-vertices of $G$ within the same time and space bounds. We implement the classical low-point algorithm of Tarjan~\cite{Tarjan72}. Recall that, the algorithm performs a DFS and computes for every vertex $v$, a value lowpoint[$v$] which is recursively defined as

\begin{align*}
\mbox{lowpoint}[v] =  \min \{ ~ DFI(v) & \cup \{\mbox{lowpoint}[s] | ~ s \mbox{ is a child of } v\} \mbox{\newline } \\
				 		         & \cup \{DFI(w) | (v,w) \mbox{ is a back-edge}\} ~ \} 
\end{align*}

Tarjan proved that if a vertex $v$ is not the root, then $v$ is a cut vertex if and only if $v$ has a child $w$ such that lowpoint[$w$] $\geq DFI(v)$. The root of a DFS tree is a cut vertex if and only if the root has more than one child. Since the lowpoint values require $\Omega(n \lg n)$ bits in the worst case, this poses the challenge of efficiently testing the condition for 
biconnectivity with $O(n)$ bits. To deal with this, as in the case of the chain decomposition algorithm, we compute lowpoint values in $O(\lg n)$ batches using our tree covering algorithm. Cut vertices encountered in the process, if at all, are stored in a separate bitmap. We show that each batch can be processed in $O(m\lg \lg n)$ time using DFS, resulting in an overall runtime of $O(m\lg n \lg \lg n)$. 

\subsubsection{Computing lowpoint and reporting cut vertices}
We first obtain a TC representation of the DFS tree using the decomposition algorithm of 
Theorem~\ref{thm:tree-decomposition} with $L = n/\lg n$. We then {\it process} 
each minitree, in the postorder of the minitrees in the minitree structure. To process a 
minitree, we compute the lowpoint values of each of the nodes in the minitree (except 
possibly the root) in overall $O(m)$ time. During the processing of any minitree, if we
determine that a vertex is a cut vertex, we store this information by marking the corresponding
node in a seperate bit vector.
Each minitree can be reconstructed in $O(m \lg\lg n)$ time using Lemma~\ref{lem:DFS-minitree}. The lowpoint value of a node is a function of the lowpoints of all its children. However the root of a minitree may have children in other minitress. Hence for the root of 
the minitree, we store the partial lowpoint value (till that point) which will be used to update its 
value when all its subtrees have computed their lowpoint values (possibly in other minitrees). 

Computing the lowpoint values in each of the minitrees is done in a two step process. In the first step, we compute and store 
the $low$ values of each node (which is the DFI value of the deepest back edge emanating from that node) belonging to the minitree using Corollary~\ref{coro}. 
Note that the $low$ values form one component of the values among which we find the minimum, in the definition of $lowpoint$ above, with a slight change. I.e., if a vertex $v$ has a backedge, then $low(v)$ is nothing but $\min \{DFI(w): (v,w)~is~a~back~edge\}$.
However, if $v$ does not have a backedge, by our convention $low(v)$ has the $DFI$ value of its parent which needs to be discounted in computing $lowpoint$ value of $v$. This is easily done if we also remember the DFI value of the parent of every node in the minitree (using $O(n)$ bits). 


Once these $low(v)$ values are computed and stored for all the vertices $v$ belonging to a minitree, they are passed on to the next step for computing $lowpoint(v)$ values. More specifically, in the second step, we 
do another DFS starting at the root of this minitree and compute the lowpoint values for all the vertices $v$ belonging to a minitree exactly in the same way as it is done in the classical Tarjan's~\cite{Tarjan72} algorithm using the explicitly stored $low(v)$ values. We provide the code snippet which actually shows how to compute lowpoint values recursively for a minitree in Algorithm $2$. Thus we obtain the following,
\begin{lemma}\label{dbe2}
Computing and storing the $lowpoint(v)$ values 
for all the nodes $v$ in a minitree can be performed in $O(m \lg \lg n)$ time, using $O(n)$ bits.
\end{lemma}

\begin{algorithm}[h]
  \begin{algorithmic}[1]
   \State{if $low(v) = DFI (parent(v))$ then}
   \State{$lowpoint(v)= DFI(v)$}
   \State{else $lowpoint(v)=Min \{ DFI(v), low(v)\}$}
   \ForAll{$y\in adj(v)$} 
         \If{$y$ is white} 
          \State{$DFI(y) \gets DFI(v)+1$}
          \State{DFS(y)} 
           \If{$lowpoint(y)<lowpoint(v)$} 
            \State{$lowpoint(v)=lowpoint(y)$}
            \EndIf
         \EndIf
   \EndFor
\end{algorithmic}
   \caption{DFS(v)}
\end{algorithm} 

To compute the effect of the roots of the minitrees on the $lowpoint$ computation, we store various 
$\Theta(\lg n)$ bit information with each of the $\Theta(\lg n)$ minitree roots including their 
partial/full lowpoint values, the rank of its first/last child in its subtree. After we process one minitree, 
we generate the next minitree, in postorder, and process it in a similar fashion and continue until 
we exhaust all the minitrees. As we can store the cut vertices in a bitvector $B$ of size $n$ marking $B[i]=1$ if and only if the vertex $v_i$ is a cut vertex, reporting them at the end of the execution of the algorithm is a routine task. 
Clearly we have taken $O(n)$ bits of space and the total running time is $O(m \lg \lg n \lg n)$ as we run the DFS algorithm  $O(\lg n)$ times overall. Thus we have the following
\begin{theorem}
Given an undirected graph $G$ with $n$ vertices and $m$ edges, in $O(m \lg n \lg \lg n)$ time and $O(n)$ bits of space we can determine whether $G$ is $2$-vertex connected. If not, in the same amount of time and space, we can compute and report all the cut vertices of $G$.
\end{theorem}


\subsection{Testing $2$-edge connectivity and finding bridges}\label{edgecon}
The classical algorithm of Tarjan~\cite{Tarjan74}
to check if $G$ is $2$-edge connected takes $O(m+n)$ time using $O(n)$ words. Schmidt's algorithm~\cite{Schmidt13} which is based on chain decomposition can also be implemented in linear time but with $O(m)$ words. The purpose of this section is to improve the space bound to $O(n)$ bits, albeit with slightly increased running time. For this, we use the following folklore characterization: a tree edge $(v,w)$, where $v$ is the parent of $w$, is a bridge if and only if lowpoint[$w$] $> DFI(v)$. That is to say, a tree edge $(v,w)$ is a bridge if and only if the vertex $w$ and any of its descedants in the DFS tree cannot reach to vertex $v$ or any of its ancestors. Thus if the edge $(v,w)$ is removed, the graph $G$ becomes disconnected. Note that, since storing the lowpoint values requires 
$\Omega(n \lg n)$ bits, we cannot store all of them at once to check the criteria mentioned in the characterization, and this poses the challenge of efficiently testing the condition for $2$-edge connectivity with only $O(n)$ bits. To perform this test in a space efficient manner, we extend ideas similar to the ones developed in the previous section.
Similar to the biconnectivity algorithm, here also we first construct a TC representation of the DFS tree using the decomposition algorithm of Theorem~\ref{thm:tree-decomposition} with $L = n/\lg n$. We then {\it process} 
each minitree, in the postorder of the minitrees in the minitree structure. To process a 
minitree, we compute the lowpoint
values of each of the nodes in the minitree (except 
possibly the root) in overall $O(m)$ time. While processing these minitrees, if we come across any bridge, we store it in a separate bitvector so that at the end of the execution of the algorithm we can report all of them.
Using Lemma~\ref{lem:DFS-minitree}, we know that each minitree can be reconstructed in $O(m \lg\lg n)$ time, and also we store for the root the partially computed lowpoint
(till the point we are done processing minitrees). 
Now we compute the lowpoint values for each of the vertices belonging to a minitree using Lemma~\ref{dbe2}. 

Once we determine lowpoint
values for all the vertices belonging to a minitree, we generate each minitree along with the node labels, and easily test whether any tree edge is a bridge using the characterization mentioned above. 
We also need to check this condition for edges that connect two minitrees; but this can also be done within the same time and space bounds. We store this information using a bit vector $B$ of length $n-1$ such that $B[i] = 1$ if and only if the $i$-th edge in pre-order, of the DFS tree, is a 
bridge. Thus, by running another DFS, we can report all the bridges of $G$. Clearly this procedure takes $O(n)$ bits of space and the total running time is $O(m \lg \lg n \lg n)$ as we run the DFS algorithm $O(\lg n)$ times overall. 
Hence we obtain the following.
\begin{theorem}
Given an undirected graph $G$ with $n$ vertices and $m$ edges, in $O(m \lg n \lg \lg n)$ time and $O(n)$ bits of space we can determine whether $G$ is $2$-edge connected. If $G$ is not $2$-edge connected, then in the same amount of time and space, we can compute and output all the bridges of $G$.
\end{theorem}

\subsection{st-numbering}\label{sec:stnumb}
The {\it st}-ordering of vertices of an undirected graph is a fundamental tool for
many graph algorithms, e.g., in planarity testing and graph drawing. The first linear-time algorithm for {\it st}-ordering the vertices of a biconnected graph is due to Even and Tarjan~\cite{EvenT76}, and is further simplified by Ebert~\cite{Ebert83}, Tarjan~\cite{Tarjan86} and Brandes~\cite{Brandes02}. All these algorithms, however, preprocess the graph using depth-first search, essentially to compute lowpoints which in turn determine an (implicit) open ear decomposition. A second traversal is required to compute the actual {\it st}-ordering. 
All of these algorithms take $O(n \lg n)$ bits of space. 
We give an $O(n)$ bits implementation of Tarjan's~\cite{Tarjan86} algorithm.

We first describe the two pass classical algorithm of Tarjan without worrying about the space requirement. The algorithm assumes, without loss of generality, that there exists an edge between the vertices $s$ and $t$, otherwise it adds the edge $(s,t)$ before starting with the algorithm. Moreover, the algorithm starts a DFS from the vertex $s$ and the edge $(s,t)$ is the first edge traversed in the DFS of $G$. Let $p(v)$ be the parent of vertex $v$ in the DFS tree. $DFI(v)$ and $lowpoint(v)$ have the usual meaning as defined previously. The first pass is a depth first search during which for every vertex $v$, $p(v)$, $DFI(v)$ and $lowpoint(v)$ are computed and stored. The second pass constructs a list $L$, which is initialized with $[s,t]$, such that if the vertices are numbered in the order in which they occur in $L$, then we obtain an {\it st}-ordering. In addition, we also have a sign array of $n$ bits, 
intialized with sign[$s$]=-. The second pass is a preorder traversal starting from the root $s$ of the DFS tree and works as described in the following pseudocode (Algorithm \ref{st}) below.

\begin{algorithm}[h]
 \begin{algorithmic}[1]
  \State{DFS(s) starts with the edge $(s,t)$}
  \ForAll{vertices $v \neq s,t $ in preorder of DFS(s)}
   \If{sign(lowpoint($v))==+$} 
          \State{Insert $v$ after $p(v)$ in $L$}                   
           \State{sign($p(v))=-$}
           \EndIf
           \If{sign(lowpoint($v$))==-} 
          \State{Insert $v$ before $p(v)$ in $L$} 
           \State{sign($p(v)$)=+}
           \EndIf
    \EndFor
      \end{algorithmic}
   \caption{st-numbering}
   \label{st}
\end{algorithm} 

It is clear from the above pseudocode that the procedure runs in linear time using $O(n \lg n)$ bits of space for storing elements in $L$. To make it space effcient, we use ideas similar to our biconnectivity algorithm. At a high level, we generate the lowpoint values of the first 
$ n/\lg n$ vertices in depth first order and process them. Due to space restriction, we cannot store the list $L$ as in Tarjan's algorithm; instead we use the BP sequence of the DFS tree and augment it with some extra information to `encode' the final {\it st}-ordering, as described below. 

As in some of our earlier algorithms, this algorithm also runs in $O(\lg n)$ phases and in each phase it processes $n/\lg n$ vertices. In the first phase, to obtain the lowpoint values of the first $n/\lg n$ vertices in depth first order, we run as in our biconnectivity algorithm a procedure to store explicitly for these vertices their lowpoint values in an array. Also during the execution of the biconnectivity algorithm, the BP sequence is generated and stored in the BP array. We create 
two more arrays, of size $n$ bits, that have one to one correspondence with the open parentheses of the BP sequence. We can use rank/select operations (as defined Section~\ref{rs}) to map the position of a vertex in these two arrays to the corresponding open parenthesis in the BP sequence. The first array, called Sign, is for storing the sign for every vertex as in Tarjan's algorithm. To simulate the effect of the list $L$, we create the second array, called $P$, where we store the relative position, i.e., ``before'' or ``after'', of every vertex with respect to its parent. Namely, if $u$ is the parent of $v$, and $v$ comes before (after, respectively) $u$ in the list $L$ in Algorithm \ref{st},
then we store $P[v]=b$ ($P[v]=a$, respectively). One crucial observation is that, even though the list $L$ is dynamic, the relative position of the vertex $v$ does not change with respect to the position of $u$, and is determined at the time of insertion of $v$ into the list $L$ (new vertices may be added between $u$ and $v$ later). In what follows, we show how to use the BP sequence, and the array $P$ to emulate the effect of list $L$ and produce the {\it st}-ordering.

We first describe how to reconstruct the list $L$ using the BP sequence and the $P$ array. The main observation 
we use in the reconstruction of $L$ is that a node $v$ appears in $L$ after all the nodes in its child subtrees whose 
roots are marked with $b$ in $P$, and also before all the nodes in its child subtrees whose roots are marked 
with $a$ in $P$. Also, all the nodes in a subtree appear ``together'' (consecutively) in the list $L$. Moreover, all the children marked $b$ appear in the increasing order of the $DFI$ while all the children marked $a$ appear in the decreasing order of the $DFI$. Thus by looking at the $P[v]$ values of all the children of a node $u$, and computing their subtree sizes, 
we can determine the position in $L$ of $u$ among all the nodes in its subtree. Let us call a child $v$ of $u$ as {\it after-child} if $v$ is marked $a$ in $P$. Similarly, if $v$ is marked $b$ in $P$, it is called {\it before-child}. Let $T(v)$ denote the subtree rooted at the vertex $v$ in the DFS tree $T$ of $G$ and $|T(v)|$ denotes the size of $T(v)$. Let us also suppose that the vertex $u$ has $k+\ell$ children, out of which $k$ children $v_1,\cdots, v_k$ are before-children and the remaining $\ell$ children $w_1,\cdots,w_\ell$ are after-children, where $DFI(v_1)< DFI(v_2)<\cdots<DFI(v_k)$ and $DFI(w_1)<DFI(w_2)<\cdots<DFI(w_\ell)$. Then in $L$, all the vertices from $T(v_1)$, $T(v_2)$, followed by till $T(v_k)$ appear, followed by $u$ and finally the vertices from $T(w_\ell)$, $T(w_{\ell-1})$ till $T(w_1)$ appear. More specifically, $u$ appears at the $(S+1)$-th location where $S=\sum_{i=1}^{k} |T(v_i)|$. With this approach, we 
can reconstruct the list $L$, and hence output the {\it st}-numbers of all the nodes in linear time, if $L$
can be stored in memory - which requires $O(n \lg n)$ bits. Now to perform this step with $O(n)$ bits, we repeat the whole process of reconstruction $\lg n$ times, 
where in the $i$-th iteration, we reproduce sublist $L[(i-1)n/\lg n + 1, \dots, i.n/\lg n]$ --
by ignoring any node that falls outside this range -- and reporting all these nodes with {\it st}-numbers in 
the range $[(i-1)n/\lg n + 1, \dots, i.n/\lg n]$. As each of these reconstruction takes $O(m \lg n \lg \lg n)$ time, we obtain the following, 

\begin{theorem}\label{thm:st-numbering}
Given an undirected biconnected graph $G$ on $n$ vertices and $m$ edges, and two distinct vertices $s$ and $t$, 
we can output an $st$-numbering of all the vertices of $G$ in $O(m \lg^2 n \lg \lg n)$ time, using $O(n)$ bits of space. 
\end{theorem}

\subsection{Applications of $st$-numbering} \label{st_app}
In this section, we show that using the space efficient implementation of Theorem~\ref{thm:st-numbering} for $st$-numbering, we immediately obtain similar results for a few applications of $st$-numbering. We provide the details below.

\subsubsection{Two-partitioning problem}\label{part}
In this problem, given vertices $a_1,\cdots, a_k$ of a graph $G$ and natural numbers $c_1,\cdots, c_k$ with $c_1+\cdots+ c_k = n$, we want to find a partition of $V$ into sets $V_1,\cdots, V_k$ with $a_i \in V_i$ and $|V_i| = c_i$ for every $i$ such that every set $V_i$ induces a connected graph in $G$. This problem is called the {\it $k$-partitioning problem}. The problem is NP-hard even when $k = 2$, $G$ is bipartite and the condition $a_i \in V_i$ is relaxed \cite{DYER85}. But, Gy\"{o}ri \cite{Gyori81} and Lovasz \cite{lovasz} proved that such a partition always exists if the input graph is $k$-connected and can be found in polynomial time in such graphs. Let G be $2$-connected. Then two-partitioning problem can be solved in the following manner~\cite{Schmidt14}: Let $v_1 := a_1$ and $v_n := a_2$, compute an $v_1v_n$-numbering $v_1, v_2, \cdots, v_n$ and note that, from the property of $st$-numbering, for any vertex $v_i$ (in particular for $i = c_1$) the graphs induced by $v_1, \cdots, v_i$ and by $v_i, \cdots , v_n$ are always connected subgraph of $G$. Thus applying Theorem \ref{thm:st-numbering}, we obtain the following:

\begin{theorem}
Given an undirected biconnected graph $G$, two distinct vertices $a_1, a_2$, and two natural numbers $c_1, c_2$ such that $c_1+c_2=n$, we can obtain a partition $(V_1, V_2)$ of the vertex set $V$ of $G$ in $O(m \lg^2 n \lg \lg n)$ time, using $O(n)$ bits of space, such that $a_1 \in V_1$ and $a_2 \in V_2$, $|V_1|=c_1$, $|V_2|=c_2$, and both $V_1$ and $V_2$ induce connected subgraph on $G$.
\end{theorem}

\subsubsection{Vertex-subset-two-partition problem}\label{subset}
Wada and Kawaguchi \cite{WadaK93} defined the following problem which they call the $vertex$-$subset$-$k$-$partition$ problem. This is actually an extension of the $k$-$partition$ problem defined in Section~\ref{part}. The problem is defined as follows: 

\noindent
{\bf Input:}
\begin{enumerate}
\item An undirected graph $G = (V, E)$ with $n$ vertices and $m$ edges;
\item a vertex subset $V' \ (\subseteq V)$ with $n' = |V'| \geq k$;
\item $k$ distinct vertices $a_i \ (1 \leq i \leq k) \in V', a_i \neq a_j \ (1 \leq i <j \leq k)$; and
\item $k$ natural numbers $n_1, n_2, \cdots, n_k$ such that 
$\sum\nolimits_{i=1}^k n_i = n'$.
\end{enumerate}

\noindent
{\bf Output:} 
a partition $V_1 \cup V_2 \cup \cdots \cup V_k$ of the vertex set $V$ and a partition $V_1' \cup V_2' \cup \cdots \cup V_k'$ of vertex set $V'$ such that for each $i(1\leq i\leq k)$
\begin{enumerate}
\item $a_i \in V_i'$;  
\item $|V_i'|=n_i$;
\item $V_i' \subseteq V_i$ and
\item each $V_i$ induces a connected subgraph of $G$.
\end{enumerate}
Note that this problem is an extension of the $k$-partition problem, since choosing $V'=V$ corresponds to the original $k$-partition problem. 
Wada and Kawaguchi~\cite{WadaK93} 
proved that vertex-subset-$k$-partition problem always admits a solution if the input graph $G$ is $k$-connected (for $k\geq 2$). In particular, if $G$ is $2$-connected, using $st$-ordering, the vertex-subset-two-partitioning problem can be solved in the following manner~\cite{WadaK93}: suppose that $G, V' \ (\subseteq V), a_1, a_2, n_1$ and $n_2 \ (n_1+n_2=n'=|V'|)$ are the inputs. Let $s = v_1 := a_1$ and $t = v_n := a_2$, compute an $st$-numbering $v_1, v_2, \cdots, v_n$. From this $st$-numbering, note 
that, $V$ now can be partitioned in two sets $V_1$ and $V_2$ such that $|V_1 \cap V'|=n_1$ and $|V_2 \cap V'|=n_2$. From the property of $st$-numbering, we know that both $V_1$ and $V_2$ induce a connected subgraph of $G$. Moreover, $a_1\in V_1$ and $a_2 \in V_2$. Using Theorem~\ref{thm:st-numbering} as a subroutine to compute such an $st$-numbering of $G$, we obtain the following result.

\begin{theorem}
Given an undirected biconnected graph $G$, we can solve the vertex-subset-two-partitioning problem in $O(m \lg^2 n \lg \lg n)$ time, using $O(n)$ bits of space.
\end{theorem}

\subsubsection{Two independent spanning trees}\label{spanning}
Recall that $k$ spanning trees of a graph $G$ are independent if they all have the same root vertex $r$, and for every vertex $v\neq r$, the paths from $v$ to $r$ in the $k$ spanning trees are vertex-disjoint (except for their endpoints). Itai and Rodeh~\cite{ItaiR88} conjectured that every $k$-connected graph contains $k$ independent spanning trees. Even though the most general version of this conjecture has not been proved yet, this conjecture is shown to be true for $k \leq 4$ \cite{CheriyanM88,CurranLY06,ItaiR88,ZehaviI89}, and also for planar graphs~\cite{Huck99}. In particular, if the given graph $G$ is biconnected, we can generate two independent spanning trees (let us call them $S$ and $T$) in the following manner~\cite{ItaiR88}.

Choose an arbitrary edge, say $(s,t)$ in $G$. Let $f$ be an $st$-numbering of $G$. To construct $S$, choose for every vertex $v \neq s$ an edge $(u,v)$ such that $f(u) <f(v)$, and for $t$ choose an edge other than $(s,t)$. To construct $T$, choose the edge $(s,t)$ and for every vertex $v \notin \{s,t\}$ an edge $(v,w)$, $f(v) <f(w)$ . It is easy to prove that $s$ is the root of both $S$ and $T$, and that $S$ and $T$ are independent spanning trees as, for every vertex $v$, the path from the root $s$ to $v$ in $S$ consists of vertices $u$ with $f(u) <f(v)$ but except the edge $(s,t)$, whereas in $T$, along with the edge $(s,t)$, it consists of vertices $w$ with $f(v) <f(w)$. Using Theorem~\ref{thm:st-numbering} to compute such an $st$-numbering of $G$, it is not hard to produce $S$ and $T$. Thus we obtain the following,

\begin{theorem}
Given an undirected biconnected graph $G$, we can report two independent spanning trees $S$ and $T$ in $O(m \lg^2 n \lg \lg n)$ time, using $O(n)$ bits.
\end{theorem}

\section{Concluding remarks and open problems}\label{conclusion}
We have presented space efficient algorithms for a number of important applications of DFS. Obtaining linear time algorithms for them while maintaining $O(n)$ bits of space usage is both interesting and challenging open problem. One of the main bottlenecks (with this approach) towards this is finding an $O(n)$-bit, $O(m+n)$-time algorithm for DFS, which is also open,
even though for BFS we know such implementations~\cite{BanerjeeC016,ElmasryHK15}. Another challenging open problem is to remove the poly-log terms in the running times of the algorithms described (e.g., the $\lg n$ term in the running time of $2$-vertex and $2$-edge connectivity algorithm, and the $\lg^2 n$ term in the running time of two independent spanning trees algorithm). These terms seem inherent in our tree covering approach. It would be interesting to find other applications of our tree covering approach for space efficient algorithms. There are also plenty of other applications of DFS, it would be interesting to study them from the point of view of space efficiency. For example, planarity is one prime example where DFS has been used very crucially. So it is a natural question that, can we test planarity of a given graph using $O(n)$ bits? \\

\noindent
{\bf\large References}

\bibliographystyle{plain}

\end{document}